\documentclass[12pt]{article}
\usepackage{amsmath,bm,amsfonts,color,amsthm, amssymb}
\usepackage[margin=2.5cm]{geometry}
\usepackage{graphicx}
\usepackage{natbib}

\usepackage[colorlinks=true,citecolor=blue,urlcolor=blue,allcolors=blue]{hyperref}

\newcommand\numberthis{\addtocounter{equation}{1}\tag{\theequation}}

\usepackage{titlesec}
\titlelabel{\thetitle.\quad}

\usepackage{tikz}
\usepackage{ctable}

\usetikzlibrary{positioning}
\usetikzlibrary{shadows}
\usepgflibrary{shapes}

\newcommand{\bb}{\boldsymbol{\beta}}

\theoremstyle{plain}
\newtheorem{theorem}{Theorem}
\newtheorem{definition}{Definition}
\newtheorem{remark}{Remark}

\usepackage{setspace}

\title{Transformed Fay-Herriot Model \\with Measurement Error in Covariates}
\author{Sepideh Mosaferi \footnote{Corresponding Author; E-mail address: \href{mailto:mosaferi@iastate.edu}{mosaferi@iastate.edu} \newline Sepideh Mosaferi is a Ph.D. Candidate with the Department of Statistics and Statistical Laboratory at Iowa State University (E-mail: \href{mailto:mosaferi@iastate.edu}{mosaferi@iastate.edu}). Malay Ghosh is a Professor with the Department of Statistics at University of Florida (E-mail: \href{mailto:ghoshm@ufl.edu}{ghoshm@ufl.edu}). Rebecca C. Steorts is an Assistant Professor with the Department of Statistical Science and Computer Science at Duke University (E-mail: \href{mailto:beka@stat.duke.edu}{beka@stat.duke.edu}).} \,\, Malay Ghosh \,\, Rebecca C. Steorts}
\date{February 16, 2021}

\begin{document}

\maketitle
\begin{abstract}
Statistical agencies are often asked to produce small area estimates (SAEs) for positively skewed variables.
When domain sample sizes are too small to support direct estimators, effects of skewness of the response variable can be large. As such, it is important to appropriately account for the distribution of the response variable given available auxiliary information. Motivated by this issue and in order to stabilize the skewness and achieve normality in the response variable, we propose an area-level log-measurement error model on the response variable. 
Then, under our proposed modeling framework, we derive an empirical Bayes (EB) predictor of positive small area quantities subject to the covariates containing measurement
error. We propose a corresponding mean squared prediction error (MSPE) of  EB predictor using both a jackknife and a bootstrap method. We show that the order of the bias is $O(m^{-1})$, where $m$ is the number of small areas. 
Finally, we investigate the performance of our methodology using both design-based  and model-based simulation studies.

\vspace{0.5em}


\end{abstract}

\section{Introduction} \label{sec:introduction}
Typically, in small area measurement error models, both the response variable and covariate can be any real number (see Ybarra and Lohr (\citeyear{ybarra2008small}), Arima et al. (\citeyear{arima2017})). However, statistical agencies are often asked to produce small area estimates (SAEs) for skewed variables, which are also positive in $\mathbb{R}^{+}$. For instance, the Census of the Governments (CoG) provides information on roads, tolls, airports, and other similar information at the local-government level as defined by the United States Census Bureau (USCB).  Another example includes the United States National Agricultural Statistics Service (NASS), which  
provides estimates regarding crop harvests (see Bellow and Lahiri (\citeyear{bellow2011empirical})). The United States Natural Resources Conservation Service (NRCS) provides estimates regarding roads at the county-level (e.g., Wang and Fuller (\citeyear{wang2003mean})), and the Australian Agricultural and Grazing Industries Survey provides estimates of the total expenditures of Australian farms (e.g., Chandra and
Chambers (\citeyear{chandra2011small})). 

When domain sample sizes are too small to support direct estimators, the effect of skewness can be quite large, and it is critical to account for the distribution of the response variable given auxiliary information at hand. For a review of the SAE literature, we refer to recent work by Rao and Molina (\citeyear{rao2015small}) and Pfefferman (\citeyear{pfeffermann2013new}).
The case of positively skewed response variables is one such that the governing parameter in the Box-Cox transformation is zero. Due to the fact that the covariate in the model may be positively skewed and contains measurement error, this has received less attention in the literature. 
Throughout this paper, we explain the problem which is beyond a simple substitution and address some of its difficulties.

\subsection{Census of the Governments}
\label{sec:cog}
As mentioned in Sec.~\ref{sec:introduction}, our proposed framework is motivated by data that is positively skewed. One such data set is the Census of Governments (CoG), which is a survey data collected by the United States Census Bureau (USCB) periodically that provides comprehensive statistics about governments and governmental activities. Data is reported on government organizations, finances, and employment. For example, data from organizations refer to location, type, and characteristics of local governments and officials. Data from finances/employment refer to revenue, expenditure, debt, assets, employees, payroll, and benefits. 
 
We utilize data from the  CoG from 2007 and 2012 (\url{https://www.census.gov/econ/overview/go0100.html}). In the CoG, the small areas consist of the 48 states of the contiguous United States. These 48 areas contain 86,152 local governments defined by the USCB, such as airports, toll roads, bridges, and other federal government corporations. The parameter of interest is the average number of full-time employees per government at the state level from the 2012 data set, which can be defined as the total number of full-time employees from all local governments divided by the total number of local governments per state. The covariate of interest is the average number of full-time employees per government at the state level from the 2007 data set.  After studying residual plots and histograms, we observe skewed patterns in the average number of full-time employees in both the 2007 and 2012 data sets, which partially motivate our proposed framework. 

\subsection{Our Contribution}
\label{sec:contribution}
Motivated by issues that statistical agencies face with skewed response variables we make several  contributions to the literature. In order to stabilize the skewness and achieve normality in the response variable, we propose an area-level log-measurement error model on the response variable (Eq. (\ref{eqn:FHlog})). In addition, we propose a log-measurement error model on the covariates (Eq. (\ref{eqn:ME})). 
Next, under our proposed modeling framework, we derive an EB predictor of positive small area quantities subject to the covariates containing measurement error. In addition, we propose a corresponding estimate of the MSPE using a jackknife and a parametric bootstrap, where we illustrate that the order of the bias is $O(m^{-1})$ under standard regularity conditions.  We illustrate the performance of our methodology in both model-based simulation and design-based simulation studies. We summarize our conclusions and provide directions for future work.

The article is organized as follows. Sec.~\ref{sec:prior work} details the prior work related to our proposed methodology. In Sec.~\ref{sec:area-level-log}, we propose a log-measurement error
 model for the response variable. In addition, we consider a measurement error model of the covariates with a log transformation.
Further, we derive the EB predictor under our framework.
Sec.~\ref{sec:mspe} provides the MSPE for our EB predictor. We provide a decomposition of the MSPE to include the uncertainty of the EB predictor through unknown parameters. Sec.~\ref{sec:estimation-mspe} provides two estimators of the MSPE, namely a jackknife and a parametric bootstrap, where we prove that the order of the bias is $O(m^{-1})$ under standard regularity conditions. Sec.~\ref{sec:experiments} provides both design-based and model-based simulation studies. Sec.~\ref{sec:dis} provides a discussion and directions for future work.

\subsection{Prior Work}
\label{sec:prior work}
In this section, we review the prior literature most relevant to our proposed work. 
There is a rich literature on the area-level Fay-Herriot model, where various additive measurement error models have been proposed on the covariates. Ybarra and Lohr (\citeyear{ybarra2008small}) proposed the first additive measurement error model on the covariates. More specifically, the authors considered covariate information from another survey that was independent of the response variable. More recently, Berg and Chandra (\citeyear{berg2014small}) have proposed an EB predictor and an approximately unbiased MSE estimator under a unit-level log-normal model, where no measurement error is assumed present in the covariates. Turning to the Bayesian literature, Arima et al. (\citeyear{arima2017}), Arima et al. (\citeyear{arima2015bayesianbook}), and Arima et al. (\citeyear{arima2015bayesian}) have provided fully Bayesian solutions to the measurement error problem for both unit-level and area-level small area estimation problems.

Next, we discuss related literature regarding the proposed jackknife and parametric bootstrap estimator of the MSPE of the Bayes estimators, where the order of the bias is $O(m^{-1}),$ under standard regularity conditions. Our proposed jackknife estimator of the MSPE contrasts that of Jiang et al. (\citeyear{jiang2002unified}), who proposed an MSE using an orthogonal decomposition, where the leading term in the MSE does not depend on the area-specific response and is nearly unbiased. Given that the authors can make an 
 orthogonal decomposition, they can show that the order of the bias of the MSE is $o(m^{-1}),$ which contrasts our proposed approach. Under our approach, the leading term 
 depends on the area-specific response, and thus, the bias is of order $O(m^{-1})$.  Turning to the bootstrap, we utilize methods similar to Butar and Lahiri (\citeyear{butar2003measures}). Using this approach, we propose a parametric bootstrap estimator of the MSPE of our estimator.  In a similar manner to the jackknife, the order of the bias for the parametric bootstrap estimator of the MSPE is $O(m^{-1})$.

\section{Area-Level Logarithmic Model with Measurement Error}
\label{sec:area-level-log}

Consider $m$ small areas and let $Y_i$ ($i=1,\ldots,m$) denote the population characteristic of interest in area $i$,
where often the information of interest is a population mean or proportion. A primary survey provides a direct estimator $y_i$ of $Y_i$ for some or all of the $m$ small areas. 
In this section, we propose a measurement error model suitable for the inference of positively skewed response variable $y_i$. To achieve normality in the response variable, we therefore propose an area-level log-measurement error model on $Y_i.$ In the rest of this section, we explain our model and the desirable predictor.

Consider the following model:
\begin{align} \label{eqn:FHlog}
z_i=\theta_i + e_i,
\end{align}
where $z_i:=\log y_i$, $\theta_i:=\log Y_i$, and $e_i$ is the sampling error distributed as $e_i \sim N(0,\psi_i)$. Assume
\begin{equation*}
\theta_i=\sum_{k=1}^{p}\beta_k \log X_{ik} + \nu_i,
\end{equation*}
where $X_{ik}$ is the $k$-th covariate of the $i$-th small area, which is unknown but is observed by $x_{ik}$. The regression coefficient $\beta_k$ is unknown and must be estimated, and $\nu_i$ is the random effect distributed as $\nu_i \sim N(0,\sigma^2_{\nu})$, where $\sigma^2_{\nu}$ is unknown.

Our measurement error model for the case of positively skewed $X_{ik}$'s is proposed as 
\begin{equation*}
w_{ik}:= \log x_{ik}=\log X_{ik} + \eta_{ik}, \qquad  k=1,...,p,
\end{equation*}
or in a vector form 
\begin{equation} \label{eqn:ME}
\boldsymbol{w}_i = \boldsymbol{W}_i + \boldsymbol{\eta}_i, \qquad \boldsymbol{\eta}_i \sim N_p(\boldsymbol{0},\Sigma_i),
\end{equation}
where $\boldsymbol{w}_i=(w_{i1}, ..., w_{ip})^\top$ and $\boldsymbol{W}_i=(W_{i1}, ..., W_{ip})^\top$ for $W_{ik}=\log X_{ik}$.
Note that in Eq. (\ref{eqn:ME}), $\boldsymbol{W}_i$ is non-stochastic within the
class of functional measurement error models (c.f. Fuller (\citeyear{fuller2006measurement})).
We assume $\Sigma_i$ is known, and if it is unknown, it can be estimated using microdata or from another independent survey. We refer to Arima et al. (\citeyear{arima2017}) for further details of estimating $\Sigma_i$.

Now, one can write
\begin{align*}
\begin{cases}
z_i=\boldsymbol{W}_i^\top \bb+ \nu_i + \bb^\top \boldsymbol{\eta}_i+ e_i \\
\theta_i= \boldsymbol{W}_i^\top \bb + \nu_i + \bb^\top \boldsymbol{\eta}_i
\end{cases}
\end{align*}
where $\bb=(\beta_1, ..., \beta_p)^\top$. Thus, for the pair $(z_i,\theta_i)$, we have the following joint normal distribution
\begin{equation*}
\begin{pmatrix}
z_i \\ \theta_i 
\end{pmatrix} \sim N_2 
\Bigg[ \begin{pmatrix}
\boldsymbol{W}_i^\top \bb \\ \boldsymbol{W}_i^\top \bb 
\end{pmatrix},
\begin{pmatrix}
\bb^\top \Sigma_i \bb+\sigma^2_{\nu}+\psi_i & \bb^\top \Sigma_i \bb + \sigma^2_{\nu}  \\ \bb^\top \Sigma_i \bb + \sigma^2_{\nu} & \bb^\top \Sigma_i \bb + \sigma^2_{\nu}
\end{pmatrix}
\Bigg].
\end{equation*}
We assume all the sources of errors $(e_i,\nu_i,\boldsymbol{\eta}_i)$ for $i=1,...,m$ are mutually independent throughout the rest of the paper.

\begin{remark}
Eq. (\ref{eqn:FHlog}) is a Fay-Herriot model for $z_i$, however, the parameter of interest is $Y_i:=\exp(\theta_i)$ rather than $\theta_i.$ Slud and Maiti (\citeyear{slud2006mean}) and Ghosh et al. (\citeyear{ghosh2015benchmarked}) used a similar model in the absence of measurement errors in the covariates.
\end{remark}

Next, we give the following conditional distribution $[\theta_i|z_i]$ to later justify our Bayesian interpretation of the unknown interested  parameter $Y_i$:
\begin{equation*}
\theta_i | z_i \sim N \Big[ \boldsymbol{W}_i^\top \bb + \frac{\bb^\top \Sigma_i \bb + \sigma^2_{\nu}}{\bb^\top \Sigma_i \bb + \sigma^2_{\nu}+\psi_i} (z_i-\boldsymbol{W}_i^\top \bb), \bb^\top \Sigma_i \bb + \sigma^2_{\nu} - \frac{(\bb^\top \Sigma_i \bb + \sigma^2_{\nu})^2}{\bb^\top \Sigma_i \bb + \sigma^2_{\nu}+\psi_i} \Big],
\end{equation*}
i.e. 
\begin{equation*}
\theta_i | z_i \sim N \Big(\gamma_i z_i + (1-\gamma_i) \boldsymbol{W}_i^\top \bb , \gamma_i \psi_i \Big),
\end{equation*}
 where $\gamma_i = (\bb^\top \Sigma_i \bb + \sigma^2_{\nu})/(\bb^\top \Sigma_i \bb + \sigma^2_{\nu}+\psi_i)$.

Recall that the parameter of interest is 
$Y_i:=\exp(\theta_i)$ after transforming from the logarithmic scale back to the original scale. Therefore, the corresponding Bayes predictor is given by $\hat{Y}_i:=E(Y_i|z_i)$. By using the moment generating function of the normal distribution of $\theta_i|z_i$, the Bayes predictor has the form of $\hat{Y}_i=\exp\{\gamma_i z_i+(1-\gamma_i)\boldsymbol{W}_i^\top \bb+\gamma_i\psi_i/2\}$.
In practice, $\boldsymbol{W}_{i}$ is unobserved, and since $E(\boldsymbol{w}_i)=\boldsymbol{W}_i$, we can replace it with the observed $\boldsymbol{w}_i$. Also, $\bb$ and $\sigma^2_{\nu}$ are unknown, and we need to replace them with their consistent estimators. Therefore, the EB predictor of $Y_i$ is
\begin{align} \label{eqn:EBpredictor}
\hat{Y}_i^{\text{EB}} = \exp\Big\{\hat{\gamma}_i z_i+(1-\hat{\gamma}_i) \boldsymbol{w}_i^\top \hat{\bb}+ \frac{\hat{\gamma}_i \psi_i}{2} \Big\}.
\end{align}

\subsection{Estimation of Unknown Parameters}
\label{sec:parameters}
In this section, we discuss estimation of the unknown parameters $\bb$ and $\sigma^2_{\nu}.$ First, an estimator of $\bb$ is obtained by solving the equation
\begin{equation} \label{eqn:bb}
\sum_{i=1}^{m} \Big[ D_i \Big(\boldsymbol{w}_i \boldsymbol{w}_i^\top - \Sigma_i \Big) \Big] \bb = \sum_{i=1}^{m} D_i \boldsymbol{w}_i z_i.
\end{equation}
The justification for Eq. (\ref{eqn:bb}) is as follows. Let $\boldsymbol{z}=(z_1,...,z_m)^\top$ and $\boldsymbol{W}^\top = (\boldsymbol{W}_1, ..., \boldsymbol{W}_m)$. Then, $\boldsymbol{z} \sim N_m(\boldsymbol{W} \bb, D^{-1})$ where $D^{-1}= \text{diag}(D_1^{-1}, ..., D_m^{-1})$ and $D_i^{-1}= \bb^\top \Sigma_i \bb + \sigma^2_{\nu}+\psi_i$. Hence, an estimator of $\bb$ is obtained by solving  
\begin{equation*}
\bb = \Big(\boldsymbol{W}^\top D \boldsymbol{W}\Big)^{-1} \boldsymbol{W}^\top D \boldsymbol{z}= \Big(\sum_{i=1}^{m} D_i \boldsymbol{W}_i \boldsymbol{W}_i^\top \Big)^{-1} \sum_{i=1}^{m} D_i \boldsymbol{W}_i z_i.
\end{equation*}

Now, notice that $E(\boldsymbol{w}_i \boldsymbol{w}_i^\top)= \boldsymbol{W}_i \boldsymbol{W}_i^\top + \Sigma_i$ and $E(\boldsymbol{w}_i)=\boldsymbol{W}_i$. Hence, we estimate $\bb$ from
\begin{equation*}
\sum_{i=1}^{m} \Big[ D_i \Big(\boldsymbol{w}_i \boldsymbol{w}_i^\top - \Sigma_i \Big) \Big] \bb = \sum_{i=1}^{m} D_i \boldsymbol{w}_i z_i.
\end{equation*}
However, $D_i$ is not known as both $\bb$ and $\sigma^2_{\nu}$ are unknown. Take $E(z_i-\boldsymbol{w}_i^\top \bb)^2 = \sigma^2_{\nu}+\psi_i$. Then $\sigma^2_{\nu}$ can be estimated from
\begin{equation} \label{eqn:sigma}
m^{-1} \sum_{i=1}^{m} \Big(z_i-\boldsymbol{w}_i^\top \bb \Big)^2 - m^{-1} \sum_{i=1}^{m} \psi_i.
\end{equation}
If the above is less than zero, estimate $\sigma_\nu^2$ as zero.
One can estimate $\bb$ and $\sigma^2_{\nu}$ by iteratively solving the Eqs. (\ref{eqn:bb}) and (\ref{eqn:sigma}).

\subsection{Mean Squared Prediction Error of the EB Predictor}
\label{sec:mspe}
In this section, we first define the MSPE of the EB predictor $\hat{Y}_i^{\text{EB}}.$ 
Second, we show that the cross-product term of the MSPE of the EB predictor $\hat{Y}_i^{\text{EB}}$ is exactly zero. Now, we introduce notation that will be used throughout the rest of the paper. Let

\begin{align*}
M_{1i} & :=E[(\hat{Y}_i-Y_i)^2|z_i] \\ 
 & =\exp\Big\{\psi_i\gamma_i\Big\}\Big[\exp\Big\{\psi_i\gamma_i\Big\}-1\Big] \exp\Big\{2\Big[\gamma_i z_i+(1-\gamma_i) \boldsymbol{W}_i^\top \bb \Big]\Big\} \\ 
M_{2i} & := E[(\hat{Y}_i^{\text{EB}}-\hat{Y}_i)^2|z_i], \quad  M_{3i}:= E[(\hat{Y}_i^{\text{EB}}-\hat{Y}_i)(\hat{Y}_i-Y_i)|z_i].
\end{align*}

Note that we estimate $\boldsymbol{W}_i$ with $\boldsymbol{w}_i$, and the term $M_{1i}$ depends on the  area-specific response variable $z_i$ unlike Jiang et al. (\citeyear{jiang2002unified}), and its estimator has bias of order $O(m^{-1})$. Since we wish to include the uncertainty of the EB predictor $\hat{Y}_i^{\text{EB}}$ with respect to the unknown parameters $\bb$ and $\sigma^2_{\nu}$, we decompose the MSPE into three terms using Definition \ref{eqn:mspe}. 

\begin{definition}
\label{eqn:mspe}
The MSPE of the EB predictor $\hat{Y}^{\text{EB}}_i$ is
\begin{align*} 
\text{MSPE}(\hat{Y}^{\text{EB}}_i) & =E[(\hat{Y}_i^{\text{EB}}-Y_i)^2|z_i]\notag \\
& \equiv E[(\hat{Y}_i-Y_i)^2|z_i] + E[(\hat{Y}_i^{\text{EB}}-\hat{Y}_i)^2|z_i] + 2 E[(\hat{Y}_i^{\text{EB}}-\hat{Y}_i)(\hat{Y}_i-Y_i)|z_i] \notag \\
& = M_{1i}+M_{2i}+2M_{3i},
\end{align*}
where we show below that $M_{3i} = 0.$
\end{definition}

To show that the cross product, $M_{3i}$ goes to 0, recall the Bayes estimator is $$E[\hat{Y}_i] = E[Y_i \mid z_i] \implies E[\hat{Y}_i-Y_i \mid z_i] = 0.$$ Consider
\begin{align*}
M_{3i} &= E[(\hat{Y}_i^{\text{EB}}-\hat{Y}_i)(\hat{Y}_i-Y_i)|z_i] \\
& = E\Big\{(\hat{Y}_i^{\text{EB}}-\hat{Y}_i) E\Big((\hat{Y}_i-Y_i)\mid z_i\Big)\Big|z_i\Big\}=0. 
\end{align*}

\section{Jackknife and Parametric Bootstrap Estimators of the MSPE}
\label{sec:estimation-mspe}
In this section, we propose two estimators for the MSPE of the EB predictor $\hat{Y}_i^{\text{EB}}.$ 
First, we propose a jackknife estimator of the MSPE. Second, we propose a parametric bootstrap estimator of the MSPE. The expectation of  the  proposed  measure  of uncertainty based on both methods  is correct up to the order $O(m^{-1})$ for the EB predictor.  

\subsection{Jackknife Estimator of the MSPE}
\label{sec:jack}
In this section, we propose a jackknife estimator of the MSPE of the  EB predictor $\hat{Y}_i^{\text{EB}},$ denoted by $\text{mspe}_J(\hat{Y}_i^{\text{EB}}).$ 
We prove the order of the bias of $\text{mspe}_J(\hat{Y}_i^{\text{EB}})$ is correct up to the order $O(m^{-1})$ under six regularity conditions. 
We propose the following jackknife estimator:

\begin{align} \label{eqn:jack-estimator}
\text{mspe}_J(\hat{Y}^{\text{EB}}_i)=\hat{M}_{1i,J}+\hat{M}_{2i,J} \quad \text{where}
\end{align}
\begin{align*}
\hat{M}_{1i,J}=
\hat{M}_{1i}-\frac{m-1}{m}\sum\limits_{j=1}^{m}(\hat{M}_{1i}-\hat{M}_{1i(-j)}) \quad
\text{and} \quad
\hat{M}_{2i,J}=\frac{m-1}{m}\sum\limits_{j=1}^{m}(\hat{Y}^{\text{EB}}_i-\hat{Y}^{\text{EB}}_{i(-j)})^2,
\end{align*}
where $(-j)$ denotes all areas except the $j$-th area. Therefore, let
\begin{align*} \label{eq:M1ihat}
\hat{M}_{1i} & :=M_{1i}(\hat{\sigma}^2_{\nu},\hat{\bb})\\
& =\exp\Big\{\psi_i\hat{\gamma}_i\Big\}\Big[\exp\Big\{\psi_i\hat{\gamma}_i\Big\}-1\Big] \exp\Big\{2\Big[\hat{\gamma}_i z_i+(1-\hat{\gamma}_i) \boldsymbol{w}_i^\top \hat{\bb}\Big]\Big\}, \numberthis   \\
\hat{M}_{1i(-j)} & = \exp\Big\{\psi_i\hat{\gamma}_{i(-j)}\Big\}\Big[\exp\Big\{\psi_i\hat{\gamma}_{i(-j)}\Big\}-1\Big] \exp\Big\{2\Big[\hat{\gamma}_{i(-j)} z_i+(1-\hat{\gamma}_{i(-j)}) \boldsymbol{w}_i^\top \hat{\bb}_{(-j)} \Big]\Big\}, \quad \text{and} \\
\hat{Y}^{\text{EB}}_{i(-j)} & = \exp \Big\{\hat{\gamma}_{i(-j)} z_i + (1-\hat{\gamma}_{i(-j)}) \boldsymbol{w}_i^\top \hat{\bb}_{(-j)} + \frac{\psi_i \hat{\gamma}_{i(-j)}}{2} \Big\}.
\end{align*}

Note that for all $[.]_{(-j)}$ cases, the $\phi=(\bb,\sigma^2_{\nu})^\top$ estimators should plug into the expressions where the data is based on all the areas other than $j$. We define some notation and then establish six regularity conditions used in Theorem \ref{thm:thm1}.
Let $\ell(\cdot|z_i)$ denote the conditional likelihood function. We define the corresponding first, second, and third derivatives of the conditional likelihood function by $\ell_i^{'}(\phi| z_i)$, $\ell_i^{''}(\phi| z_i),$ and $\ell_i^{'''}(\phi| z_i)$, respectively.
Now, assume the following six regularity conditions: 

\noindent\textit{Condition 1.} Define $\phi^\top=(\bb,\sigma^2_{\nu}) \in \Theta$ where $\Theta$ is a compact set such that $\Theta \subseteq (\mathbb{R}^{p},\mathbb{R}^{+})$.

\noindent\textit{Condition 2.} Assume $\hat{\phi}$ is a consistent estimator for $\phi,$ i.e. $\hat{\phi} \xrightarrow{p} \phi$.

\noindent\textit{Condition 3.} Assume  $\ell_i^{'}(\phi| z_i)$ and $\ell_i^{''}(\phi| z_i)$ both exist for $i = 1,\ldots m$, almost surely in probability.

\noindent\textit{Condition 4.} Assume $E\{\ell_i^{'}(\phi|z_i)| \phi\}=0$ for $i = 1,\ldots, m$.

\noindent\textit{Condition 5.} Assume $\ell_i^{''}(\phi| z_i)$ is a continuous function of $\phi$ for $i = 1, ..., m$, almost surely in probability, where $E\{\ell_i^{''}(\phi|z_i)\}$ is positive definite, uniformly bounded away from $0$, and is a measurable function of $z_i$.

\noindent\textit{Condition 6.} Assume $E\{|\ell_i^{'}(\phi|z_i)|^{4+\delta}\}$, $E\{|\ell_i^{''}(\phi|z_i)|^{4+\delta}\}$, and  $E\{\sup_{c \in (-\epsilon,\epsilon)}|\ell_i^{'''}(\phi+c|z_i)|^{4+\delta}\}$ are uniformly bounded for $i = 1,\ldots m$ under some $\epsilon > 0$ and $\delta > 0$.

\begin{theorem}
\label{thm:thm1}
Assume Conditions 1--6 hold. Then
$$E[\text{mspe}_J(\hat{Y}_i^{\text{EB}})]=\text{MSPE}(\hat{Y}_i^{\text{EB}})+O(m^{-1}).$$
\end{theorem}

\begin{proof}
Define

\begin{align*}
E(\text{mspe}_J(\hat{Y}_i^{\text{EB}})) & \equiv E(\hat{M}_{1i,J}+\hat{M}_{2i,J}) \\
& = E \Big(\hat{M}_{1i}-\frac{m-1}{m} \sum_{j=1}^{m} [(\hat{M}_{1i}-\hat{M}_{1i(-j)})|z_i] \Big) \\
& \quad + \frac{m-1}{m} E \Big(\sum_{j=1}^{m} [(\hat{Y}_i^{\text{EB}}-\hat{Y}_{i(-j)}^{\text{EB}})^2|z_i] \Big).
\end{align*}

\noindent Also, define a remainder term $r_i$ that is bounded in absolute value by $R_i$ such that, 

$$|r_i| \leq \max\{1,|\ell^{'}(\phi|z_i)|^3,|\ell^{''}(\phi|z_i)|^3,|\ell^{'''}(\phi|z_i)|^3\} \equiv R_i.$$

First, we prove $\hat{M}_{1i,J}$ has a bias of order $O(m^{-1}).$
Using a Taylor series expansion, we find that

\begin{align*}
\hat{M}_{1i}=M_{1i}+M_{1i}^{' \top}(\phi) (\hat{\phi}-\phi) + \frac{1}{2} M_{1i}^{'' \top}(\phi) (\hat{\phi}-\phi)^2 + \frac{1}{6} M_{1i}^{''' \top}(\phi^*) (\hat{\phi}-\phi)^3,
\end{align*}
for $\phi^*$ between $\phi$ and $\hat{\phi}$. Also, $M_{1i}^{' \top}(\phi)$, $M_{1i}^{'' \top}(\phi)$, and $M_{1i}^{''' \top}(\phi^*)$ stand for the first, second, and third derivatives of $M_{1i}$ with respect to $\phi$. 
Let $\hat{\phi}^{\top}~=~(\hat{\bb},\hat{\sigma}^2_{\nu})$, and it follows that 
\begin{align*}
\hat{M}_{1i}-\hat{M}_{1i(-j)}=\hat{M}^{' \top}_{1i}(\hat{\phi})(\hat{\phi}-\hat{\phi}_{(-j)})+\frac{1}{2} \hat{M}_{1i}^{'' \top}(\hat{\phi}) (\hat{\phi}-\hat{\phi}_{(-j)})^2+\frac{1}{6} \hat{M}_{1i}^{''' \top}(\hat{\phi}^*_{(-j)}) (\hat{\phi}-\hat{\phi}_{(-j)})^3,
\end{align*}
for some $\hat{\phi}^*_{(-j)}$ between $\hat{\phi}_{(-j)}$ and $\hat{\phi}$.
In order to approximate the solution $\hat{\phi}$ of the equation $f(\tau)=\sum_{i=1}^{m}\ell'(\tau|z_i)=0$ in iteration $(\xi+1)$, we use Householder's method (Householder (\citeyear{householder1970numerical}), Theorem 4.4.1). See also Theorem 1 of Lohr and Rao (\citeyear{lohr2009jackknife}):
\begin{align*}
\tau_{\xi+1}= \tau_{\xi}-\frac{f(\tau_{\xi})}{f'(\tau_{\xi})} \Big[1+\frac{\tau_{\xi} f''(\tau_{\xi})}{2 \{f'(\tau_{\xi})\}^2} \Big].
\end{align*}
By taking the initial value $\tau_{\xi}=\phi$, we have

\begin{align*}
\hat{\phi}-\phi & =-\frac{\sum\limits_{i=1}^{m}\ell_i^{'}(\phi|z_i)}{\sum\limits_{i=1}^{m} \ell_i^{''}(\phi|z_i)}\Bigg\{1+\frac{\sum\limits_{k=1}^{m}\ell^{'}_{k}(\phi| z_k) \sum\limits_{r=1}^{m}\ell^{'''}_{r}(\phi|z_r)}{2(\sum\limits_{k=1}^{m}\ell^{''}_{k}(\phi| z_k))^2}\Bigg\}+O_p(|\hat{\phi}-\phi|^3), \quad \text{and} \\
\hat{\phi}-\hat{\phi}_{(-j)} & =\frac{\ell^{'}_j(\hat{\phi}|z_j)}{\sum\limits_{k \neq j}^{m}\ell^{''}_k(\hat{\phi}|z_k)} \Bigg[ 1-\frac{\ell^{'}_j(\hat{\phi}|z_j) \sum\limits_{k \neq j}^{m} \ell^{'''}_k(\hat{\phi}| z_k)}{2( \sum\limits_{k \neq j}^{m} \ell^{''} _k(\hat{\phi}| z_k) )^2} \Bigg]+ O_p(|\hat{\phi}-\hat{\phi}_{(-j)}|^3).
\end{align*}

By taking conditional expectation and using Theorem 2.1 of Jiang et al. (\citeyear{jiang2002unified}), we find that
\begin{align*}
E(\hat{\phi}-\phi|z_i)=\frac{-\ell_i^{'}(\phi|z_i)+\varphi}{\sum\limits_{i=1}^{m}E\{\ell^{''}_i(\phi|z_i)\}}+r_i \, o(m^{-1}),
\end{align*}
where 
\begin{align*}
\varphi= \frac{\sum\limits_{j=1}^{m}E[\ell^{'}_j(\phi|z_j) \ell^{''}_j(\phi|z_j)]}{\sum\limits_{j=1}^{m}E\{\ell^{''}_j(\phi|z_j)\}}-\frac{\sum\limits_{j=1}^{m}\sum\limits_{k=1}^{m}E{[\ell^{'}_j(\phi|z_j)]^2}E(\ell^{'''}_k(\phi|z_k))}{2(\sum\limits_{j=1}^{m}E\{\ell^{''}_j(\phi|z_j)\})^2},
\end{align*}

\begin{align*}
\sum\limits_{j \neq i}^{m}E(\hat{\phi}-\hat{\phi}_{(-j)}|z_i) & =\frac{-\ell^{'}_i(\phi|z_i)+\varphi}{\sum\limits_{j=1}^{m}E\{\ell^{''}_j(\phi|z_j)\}}+r_i \, o(m^{-1}), \quad \text{and} \\
E(\hat{\phi}_{(-i)}-\hat{\phi}|z_i) & = \frac{\ell^{'}_i(\phi|z_i)}{\sum\limits_{j=1}^{m}E\{\ell^{''}_j(\phi|z_j)\}}+r_i \, o(m^{-1}).
\end{align*}

By combining the above results, we find that
\begin{align*}
E\{\hat{M}_{1i,J}-M_{1i}|z_i\} &= -M_{1i}^{'}(\phi|z_i) \ell^{'}_i(\phi|z_i)/ \varphi +r_i \, o(m^{-1}).\\
 \text{Hence,} \quad E(\hat{M}_{1i,J}) & = M_{1i}+O(m^{-1}).
\end{align*}

Second, we prove $\hat{M}_{2i}$ has a bias of order $o(m^{-1})$. Let
\begin{align*}
\hat{Y}_i^{\text{EB}}-\hat{Y}_{i(-j)}^{\text{EB}} := h(\hat{\phi}|z_i)-h(\hat{\phi}_{(-j)}|z_i),
\end{align*}
and $h(\phi|z_i)=E(Y_i^{\text{EB}}|z_i,\phi)$. Using a Taylor series expansion, we find that
\begin{align*}
\hat{Y}_i^{\text{EB}}-\hat{Y}_{i(-j)}^{\text{EB}}=h^{' \top} (\hat{\phi}|z_i) (\hat{\phi}-\hat{\phi}_{(-j)})+\frac{1}{2}h^{'' \top}(\hat{\phi}^*_{(-j)}|z_i) (\hat{\phi}-\hat{\phi}_{(-j)})^2,
\end{align*}
where 
\begin{align*}
h^{' \top}(\hat{\phi}|z_i)=\Big(\frac{\partial h(\hat{\phi}|z_i)}{\partial \bb}, \frac{\partial h(\hat{\phi}|z_i)}{\partial \sigma^2_{\nu}}\Big), \quad h^{'' \top}(\hat{\phi}|z_i)=\Big(\frac{\partial(\partial h(\hat{\phi}|z_i))}{\partial^2 \bb}, \frac{\partial(\partial h(\hat{\phi}|z_i))}{\partial^2 \sigma^2_{\nu}}\Big),
\end{align*}
and $\hat{\phi}^*_{(-j)}$ is between $\hat{\phi}_{(-j)}$ and $\hat{\phi}.$
Using an additional Taylor series expansion, we find that

\begin{align*}
\sum\limits_{j=1}^{m}E\big\{(\hat{Y}_{i(-j)}^{\text{EB}}-\hat{Y}_i^{\text{EB}})^2|z_i\big\} & =\big\{h^{' \top}(\phi|z_i)\big\}^2 \times \frac{\sum\limits_{j=1}^{m}E\{(\ell^{'}_j(\phi|z_j))^2\}}{\varphi^2}+r_i \, o(m^{-1}). \quad \text{Similarly,}\\
E\big\{(\hat{Y}_{i}^{\text{EB}}-\hat{Y}_i)^2|z_i\big\} & =\big\{h^{' \top}(\phi|z_i)\big\}^2 \times \frac{\sum\limits_{j=1}^{m}E\{(\ell^{'}_j(\phi|z_j))^2\}}{\varphi^2}+r_i \, o(m^{-1}).
\end{align*}
By combining the above results, we find that
\begin{align*}
 E(\hat{M}_{2i,J}) &= M_{2i}+o(m^{-1}).
\end{align*}

Finally,
\begin{align*}
E(\text{mspe}_J(\hat{Y}_i^{\text{EB}})) & = E(\hat{M}_{1i,J})+E(\hat{M}_{2i,J}) \\
& = \{M_{1i}+O(m^{-1}) \} + \{ M_{2i}+o(m^{-1}) \} \\
& = M_{1i}+M_{2i} + O(m^{-1}).\\
\end{align*}
Hence, $E[\text{mspe}_J(\hat{Y}_i^{\text{EB}})] =\text{MSPE}(\hat{Y}_i^{\text{EB}})+O(m^{-1})$.

\end{proof}

\subsection{Parametric Bootstrap Estimator of the MSPE}
\label{sec:boot}
In this section, we propose a parametric bootstrap estimator of the MSPE of the EB predictor $\hat{Y}_i^{\text{EB}}$, which we denote it by $\text{mspe}_B (\hat{Y}_i^{\text{EB}}).$ We prove that the order of the bias is correct up to order $O(m^{-1}).$ Specifically, we extend Butar and Lahiri (\citeyear{butar2003measures}) to find a parametric bootstrap of our proposed EB predictor.
To introduce the parametric bootstrap method, consider the following bootstrap model: 
\begin{align*}
\label{eqn:Bootmodel}
z_i^{\star}|\boldsymbol{w}_i^{\star}, \nu_i^{\star} & \stackrel{ind}{\sim} N(\boldsymbol{w}_i^{{\star} \top} \hat{\bb} +\nu_i^{\star}, \psi_i)\\
\boldsymbol{w}_i^{\star} & \stackrel{ind}{\sim} N_p(\boldsymbol{W}_i, \Sigma_i) \\
\nu_i^{\star} & \stackrel{ind}{\sim} N(0, \hat{\sigma}^2_{\nu}). \numberthis
\end{align*}

Recall that from Definition \ref{eqn:mspe}, $\text{MSPE}(\hat{Y}_i^{\text{EB}}) = M_{1i} + E[(\hat{Y}_i^{\text{EB}} - \hat{Y}_i)^2|z_i]$ since $M_{3i}=0$.
We use the parametric bootstrap twice. First, we use it to estimate $M_{1i}$ in order to correct the bias of $\hat{M}_{1i}:=M_{1i}(\hat{\sigma}^2_{\nu}, \hat{\boldsymbol{\beta}})$ (see Eq. (\ref{eq:M1ihat})). Second, we use it to estimate $E[(\hat{Y}_i^{\text{EB}}-\hat{Y}_i)^2|z_i]$. More specifically,
we propose to estimate $M_{1i}$ by 
$2{M}_{1i}(\hat{\sigma}^2_{\nu},\hat{\boldsymbol{\beta}})-E_{\star}[M_{1i}(\hat{\sigma}^{\star 2}_{\nu},\hat{\boldsymbol{\beta}}^{\star})|z_i^{\star}]$,
and $E[(\hat{Y}_i^{\text{EB}}-\hat{Y}_i)^2|z_i]$ by $E_{\star}[(\hat{Y}_i^{\text{EB} \star}-\hat{Y}_i^{\text{EB}})^2|z_i^{\star}]$, where $E_{\star}$ denotes  that the  expectation  is computed  with  respect  to model in Eq. (\ref{eqn:Bootmodel}) and $\hat{Y}_i^{\text{EB} \star}=\exp \{\hat{\gamma}_i^{\star}z_i+(1-\hat{\gamma}_i^{ \star}) \boldsymbol{w}_i^\top \hat{\bb}^{\star} +\psi_i \hat{\gamma}_i^{ \star}/2 \}$. 
In addition, $\hat{\gamma}_i^{ \star}=(\hat{\sigma}^{\star 2}_{\nu}+\hat{\bb}^{\star \top} \Sigma_i \hat{\bb}^{\star})/(\hat{\sigma}^{\star 2}_{\nu}+\hat{\bb}^{\star \top} \Sigma_i \hat{\bb}^{\star}+\psi_i),$ where $\hat{\bb}^{\star}$ and $\hat{\sigma}^{\star 2}_{\nu}$ are estimators of $\bb$ and $\sigma^2_{\nu}$ with respect to the parametric bootstrap model in Eq. (\ref{eqn:Bootmodel}).

Our proposed estimator of $\text{MSPE}(\hat{Y}_i^{\text{EB}})$ is

\begin{equation} \label{eq:3.4}
\text{mspe}_B(\hat{Y}_i^\text{EB}) = 2{M}_{1i}(\hat{\sigma}^2_{\nu},\hat{\boldsymbol{\beta}})-E_{\star}[M_{1i}(\hat{\sigma}^{\star 2}_{\nu},\hat{\boldsymbol{\beta}}^{\star})|z_i^{\star}]+E_{\star}[(\hat{Y}_i^{\text{EB} \star}-\hat{Y}_i^{\text{EB}})^2|z_i^{\star}],
\end{equation}
which has bias of order $O(m^{-1})$ as shown in the Theorem \ref{thm:thm2}.
\begin{theorem}
\label{thm:thm2}
Assume $E_{\star}(\hat{\sigma}_{\nu}^{\star 2}-\hat{\sigma}^2_{\nu})=O_p(m^{-1})$ and $E_{\star}(\hat{\boldsymbol{\beta}}^{\star}-\hat{\boldsymbol{\beta}})=O_p(m^{-1})$. The bootstrap estimator of the MSPE has bias of order $O(m^{-1})$, i.e.
$$E[\text{mspe}_B(\hat{Y}_i^{\text{EB}})]= \text{MSPE}(\hat{Y}_i^{\text{EB}}) + O(m^{-1}).$$
\end{theorem}
\begin{proof}
Let
\begin{equation*}
E_{\star}[M_{1i}(\hat{\sigma}^{\star 2}_{\nu},\hat{\boldsymbol{\beta}}^{\star})|z_i^{\star}] = M_{1i}(\hat{\sigma}^2_{\nu},\hat{\boldsymbol{\beta}}) + O_p(m^{-1}).
\end{equation*}
Assume that $R^{{\star}}_m=O_{p^{\star}}(m^{-1})$ such that $mR^{\star}_m$ is bounded in probability under the parametric bootstrap model in Eq. (\ref{eqn:Bootmodel}). Consider the following Taylor series expansion:

\begin{equation*}
\hat{Y}_i^{\text{EB} \star} - \hat{Y}_i^{\text{EB}} = (\hat{\phi}^{\star} - \hat{\phi})^\top h'(\hat{\phi}| z_i) + R^{\star}_m,
\end{equation*}
such that $\hat{\phi}^{\star \top}=(\hat{\bb}^{\star},\hat{\sigma}^{\star 2}_{\nu})$.

Using an argument similar to the proof of Theorem \ref{thm:thm1}, 
\begin{align}
\label{eqn:star}
E_{\star}[(\hat{Y}_i^{\text{EB} \star}-\hat{Y}_i^{\text{EB}})^2|z_i^{\star}]&=\hat{M}_{2i}+o_p(m^{-1}) \quad \text{and} \quad
E_{\star}[\hat{M}^{\star}_{1i}|z_i^{\star}] =\hat{M}_{1i}+O_p(m^{-1}).
\end{align}
Substituting Eq. (\ref{eqn:star}) into Eq. (\ref{eq:3.4}), we find that
\begin{align*}
\text{mspe}_B(\hat{Y}_i^{\text{EB}})& =2\hat{M}_{1i}-[\hat{M}_{1i}+O_p(m^{-1})]+\hat{M}_{2i}+o_p(m^{-1}) \\
& = \hat{M}_{1i}+\hat{M}_{2i}+O_p(m^{-1}).
\end{align*}
This suggests that
$$ E[\text{mspe}_B(\hat{Y}_i^{\text{EB}})]= \text{MSPE}(\hat{Y}_i^{\text{EB}})+O(m^{-1}). $$
\end{proof}

\section{Experiments}
\label{sec:experiments}
In this section, we investigate the performance of the EB predictors in comparison to the direct estimators through design-based and model-based simulation studies. In addition, we evaluate the MSPE estimators using both a jackknife and parametric bootstrap. 


\subsection{Design-Based Simulation Study}
\label{sec:design}
In this section, we consider a design-based simulation study using the CoG data set as described in Sec. \ref{sec:cog}.

\subsubsection{Design-Based Simulation Setup}
We describe the design-based simulation setup. The parameter of interest is average number of full-time employees per government at the state level from 2012 data set. The covariate
is the average number of full-time employees per government at the state level from the 2007 data set. There are observed skewed patterns in the average number of full-time employees in both 2007 and 2012, which motivates our proposed framework.

For the response variable, we select a total sample of 7,000 governmental units proportionally allocated to the states and for the covariates, we select a total sample of 70,000 units and the survey-weighted averages were then calculated. The measurement error variance $\Sigma_i$ was obtained from a Taylor series approximation, where $\text{Var}(x_i)$ was estimated from the formula of variance in simple random sampling without replacement at each state. The $\psi_i$'s were estimated by a Generalized Variance Function (GVF) method (see  Fay and Herriot (\citeyear{fay1979estimates})). We assume the sampling variances to be known throughout the estimation procedure. 

For the design-based simulation, we draw 1,000 samples and estimate the  parameters from each sample. We evaluate our proposed predictors by empirical MSE per each state $i$:
\begin{align*}
\text{EMSE}(\hat{Y}_i) =\frac{1}{R}\sum\limits_{r=1}^{R}\Big[\hat{Y}_i^{(r)}-Y_i^{(r)}\Big]^2,
\end{align*}  
where $R=1,000$ is the total number of replications, and $\hat{Y}_i$ is the estimator of $Y_i$. 
In addition, when the parametric bootstrap it used, we take $B=1,000$ bootstrap samples. We use the same number of replications and bootstrap samples in the design and model-based simulation studies.

\subsubsection{Design-Based Simulation Results}
In this section, we provide the results of the design-based simulation study. 

\paragraph{Investigating the performance of the proposed estimators}

Recall that the covariate of interest is the average number of full-time employees per government at the state level from 2007 data set, and we wish to predict the average number of full-time employees per government at the state level in 2012. To do so, we give the predictors for each state as well as their corresponding EMSE's in Tables \ref{table1} and \ref{table2}. More specifically, we compare the following three estimators:
\begin{itemize}
\item[1)] $y_i$: the direct estimator,
\item[2)] $\tilde{Y}_i$: the EB predictor, assuming the true covariate $w_i$ and ignoring $\Sigma_i$ in our model,
\item[3)] $\hat{Y}_i^{\text{EB}}$: the EB predictor, assuming the true covariate $w_i$ has measurement error, where $\Sigma_i$ is included in our model. 
\end{itemize}

We observe that in most cases the $\text{EMSE}(\hat{Y}_i^{\text{EB}})$ is smaller than the $\text{EMSE}(\tilde{Y}_i)$. However,  we observe that our proposed EB predictor does not always outperform the direct estimator, which we further explore in our model-based simulation studies in Sec.~\ref{sec:model}.

\begin{table}[ht]
\centering
\caption{Estimators and their empirical MSEs from CoG. Note that $n_i$ is the sample size per state, and the MSEs are rescaled logarithmically.}
\label{table1}
\renewcommand{\arraystretch}{1.25}
\vspace{0.25cm}
\begin{tabular}{@{} ccccccccccc @{}}
\hline
$i$ & State & $n_i$ & $y_i$ & $\tilde{Y}_i$ & $\hat{Y}_i^{\text{EB}}$ & $\text{EMSE}(y_i)$ & $\text{EMSE}(\tilde{Y}_i)$ & $\text{EMSE}(\hat{Y}_i^{\text{EB}})$  \\ \specialrule{.1em}{.05em}{.05em} 
1 & RI & 10 & 191.641 & 202.928 & 204.907 & 6.523 & 5.390 & 5.103 \\
2 & AK & 14 & 132.301 & 120.112 & 123.684 & 4.501 & 6.153 & 5.793 \\
3 & NV & 15 & 420.299 & 422.992 & 431.912 & 8.077 & 8.170 & 8.449 \\
4 & MD & 19 & 824.939 & 784.779 & 794.784 & 8.050 & 5.524 & 6.503 \\
5 & DE & 27 & 64.273 & 72.674 & 73.130 & 3.003 & 5.113 & 5.182 \\
6 & LA & 40 & 363.838 & 342.056 & 343.344 & 6.017 & 0.845 & -2.863 \\
7 & VA & 40 & 560.881 & 526.612 & 530.160 & 6.669 & 8.265 & 8.148 \\
8 & NH & 44 & 80.695 & 83.645 & 83.857 & -3.994 & 2.254 & 2.387 \\
9 & UT & 47 & 126.045 & 117.729 & 118.512 & 2.827 & 2.873 &  2.461 \\
10 & AZ & 49 & 332.610 & 349.704 & 350.915 & 6.270 & 7.380 &  7.439 \\
11 & CT & 49 & 187.500 & 191.494 & 192.054 & 4.851 & 5.457 & 5.528 \\
12 & SC & 53 & 231.790 & 221.881 & 222.574 & 5.158 & 6.279 & 6.218 \\
13 & WV & 53 & 83.568 & 84.071 & 84.315 & 2.754 & 2.482 &  2.336 \\
14 & WY & 55 & 43.084 & 39.629 & 39.795 & 2.493 & 3.873 & 3.824 \\
15 & VT & 59 & 27.717 & 27.529 & 27.574 & 0.474 & 0.751 &  0.688 \\
16 & ME & 65 & 38.892 & 40.713 & 40.789 & 2.966 & 1.899 & 1.840 \\
17 & NM & 67 & 93.817 & 93.360 & 93.913 & 3.496 & 3.331 & 3.529 \\
18 & MA & 70 & 237.066 & 231.213 & 231.999 & 4.218 & 1.741 & 2.310 \\
19 & TN & 72 & 245.317 & 232.133 & 233.405 & 4.257 & 6.144 & 6.023 \\
20 & NC & 76 & 372.264 & 357.791 & 359.375 & 5.846 & 6.997 & 6.700 \\
21 & MS & 78 & 137.977 & 132.143 & 132.454 & 3.891 & 0.299 & 0.774 \\
22 & ID & 92 & 46.450 & 45.642 & 45.784 & 2.451 & 1.910 & 2.016 \\
23 & AL & 93 & 162.841 & 160.389 & 160.818 & 3.356 & 2.131 & 2.406 \\
24 & MT & 99 & 23.296 & 22.457 & 22.508 & 0.461 & 1.481 & 1.433 \\
25 & KY & 104 & 119.415 & 121.105 & 121.576 & 1.935 & 2.928 & 3.134 \\
26 & NJ & 108 & 235.215 & 245.163 & 245.595 & 3.898 & 5.663 & 5.713 \\
27 & GA & 109 & 255.141 & 243.862 & 244.539 & 5.686 & 6.696 & 6.648 \\

\hline

\end{tabular}
\end{table}

\begin{table}[ht]
\centering
\caption{Estimators and their empirical MSEs from CoG. Note that $n_i$ is the sample size per state, and the MSEs are rescaled logarithmically (continued).} 
\label{table2}
\renewcommand{\arraystretch}{1.25}
\vspace{0.25cm}
\begin{tabular}{@{} ccccccccccc @{}}
\hline
$i$ & State & $n_i$ & $y_i$ & $\tilde{Y}_i$ & $\hat{Y}_i^{\text{EB}}$ & $\text{EMSE}(y_i)$ & $\text{EMSE}(\tilde{Y}_i)$ & $\text{EMSE}(\hat{Y}_i^{\text{EB}})$ \\ \specialrule{.1em}{.05em}{.05em} 
28 & AR & 118 & 68.911 & 65.223 & 65.362 & -3.160 & 2.719 &  2.646 \\
29 & OR & 120 & 72.363 & 72.483 & 72.653 & 1.376 & 1.493 & 1.648 \\
30 & FL & 124 & 377.822 & 365.035 & 366.658 & 7.658 & 8.148 & 8.092 \\
31 & OK & 144 & 76.686 & 74.610 & 74.803 & -2.447 & 1.155 & 0.926 \\
32 & WA & 145 & 93.681 & 91.238 & 91.476 & 3.077 & 3.920 & 3.852 \\
33 & SD & 153 & 14.757 & 14.078 & 14.142 & -1.338 & -3.583 & -4.546 \\
34 & IA & 155 & 56.219 & 55.686 & 55.790 & -4.924 & -0.962 & -1.332 \\
35 & CO & 184 & 74.373 & 71.579 & 71.908 & 0.789 & 2.907 & 2.746 \\
36 & NE & 197 & 33.713 & 31.017 & 31.215 & 1.326 & -0.561 & -1.167 \\
37 & IN & 217 & 75.994 & 78.712 & 78.830 & 0.619 & 0.609 & 0.775 \\
38 & ND & 217 & 8.336 & 7.421 & 7.469 & -1.704 & -1.436 & -1.644 \\
39 & MI & 236 & 85.314 & 97.710 & 97.705 & -1.220 & 4.945 & 4.944 \\
40 & WI & 246 & 61.304 & 61.320 & 61.451 & 2.486 & 2.495 & 2.569 \\
41 & NY & 270 & 280.982 & 278.784 & 283.973 & 6.457 & 6.275 & 6.681 \\
42 & MO & 278 & 61.972 & 62.846 & 62.946 &  0.093 & 1.307 & 1.409 \\
43 & MN & 289 & 42.025 & 45.747 & 45.727 & -1.572 & 2.367 & 2.355 \\
44 & OH & 302 & 108.261 & 111.756 & 111.905 & 2.364 & 3.820 & 3.864 \\
45 & KS & 309 & 30.804 & 30.253 & 30.321 & 1.859 & 2.253 & 2.208 \\
46 & CA & 338 & 238.284 & 248.558 & 248.800 & 6.991 & 6.244 & 6.223 \\
47 & TX & 374 & 221.105 & 204.983 & 205.689 & 2.570 & 5.965 & 5.892 \\
48 & PA & 386 & 81.653 & 83.551 & 83.668 & 2.888 & 3.628 & 3.666 \\
49 & IL & 567 & 64.650 & 67.278 & 67.172 & 1.490 & 3.110 & 3.064 \\
\hline
 \vspace{0.25cm}

\end{tabular}
\end{table}

\paragraph{Jackknife versus parametric bootstrap estimators}
Next, we consider the performance of MSPE estimators, i.e., the  jackknife and bootstrap, with respect to the true MSE, i.e., $\text{EMSE}(\hat{Y}_i^{\text{EB}})$ in Figure \ref{figure1}. The results are given on the logarithmic scale, and we observe that the distribution of jackknife is closer to the distribution of true MSE when compared to the bootstrap.
Therefore, we recommend the jackknife given that it slightly overestimates the true MSE.
As already mentioned, given that our proposed estimator does not uniformly beat the direct estimator in terms of the EMSE, we conduct a model-based simulation study in Sec.~\ref{sec:model} to investigate this and provide further insight.

\begin{figure}[h!]
\centering
\begin{tabular}{c}  
 \includegraphics[width=0.8\textwidth]{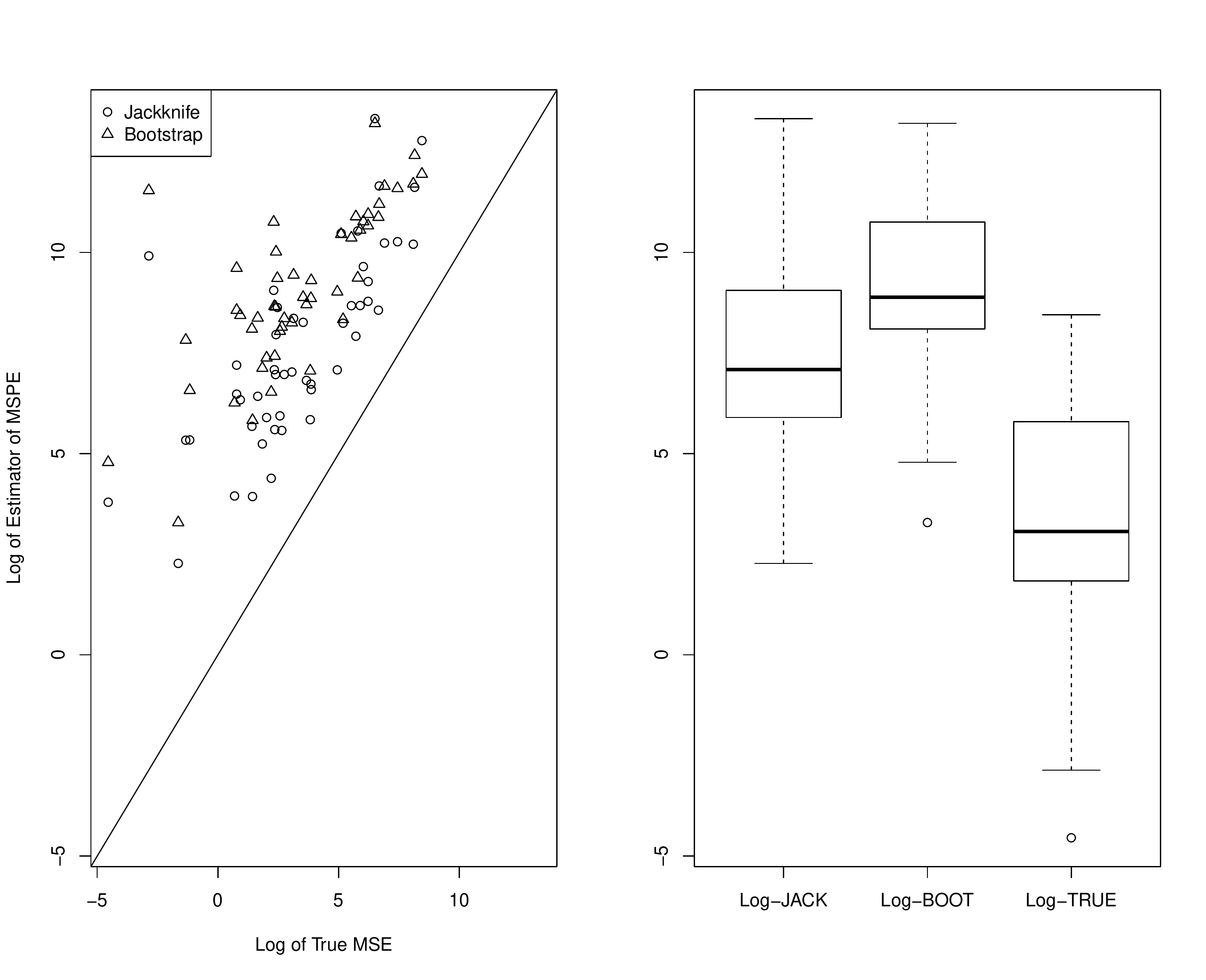} 
\end{tabular}
\caption{Left: The jackknife and the bootstrap estimators versus the true MSE ($\text{EMSE}(\hat{Y}_i^{\text{EB}})$), where the results are rescaled logarithmically. Right: Box plots of the jackknife and the bootstrap estimators and the true MSE ($\text{EMSE}(\hat{Y}_i^{\text{EB}})$), where the results are rescaled logarithmically. In
general, the log of the jackknife is closer to the log of true MSE.}

\label{figure1}
\end{figure}

\subsection{Model-Based Simulation Study}
\label{sec:model}
In this section, we describe our model-based simulation study to further investigate the performance of the proposed  EB predictor $\hat{Y}_i^{\text{EB}}.$ Second, we compare the proposed jackknife and parametric bootstrap estimators, $\text{mspe}_J (\hat{Y}_i^{\text{EB}})$ and $\text{mspe}_B (\hat{Y}_i^{\text{EB}}).$
Third, we investigate how often the variance estimates $\hat{\sigma}_u^2$ are zero. Finally, we investigate how the regression parameter changes when $\Sigma_i$ is misspecified.
Our goal through this model-based simulation study is to understand how one could improve the EB predictor through future research, and to further understand its underlying behavior.

\subsubsection{Model-Based Simulation Setup} 
\label{sec:setup}
In this section, we provide the setup of our model-based simulation study in Table \ref{table3}. This setup follows Eqs. (\ref{eqn:FHlog}) and (\ref{eqn:ME}).
We are interested in comparing the following four estimators:
\begin{itemize}
\item[1)] $y_i$: the direct estimator,
\item[2)] $\hat{Y}_i$: the EB predictor, assuming the true covariate $W_i$
\item[3)] $\tilde{Y}_i$: the EB predictor, assuming the true covariate $w_i$ and ignoring $\Sigma_i$ in our model,
\item[4)] $\hat{Y}_i^{\text{EB}}$: the EB predictor, assuming the true covariate $w_i$ has measurement error, where $\Sigma_i$ is included in our model 
\end{itemize}

We compare these four estimators (for each area $i$) using the empirical MSE: 
$$\text{EMSE}(\hat{Y}_i) = \frac{1}{R}\sum\limits_{r=1}^{R}\Big[\hat{Y}_{i}^{(r)}-Y_{i}^{(r)}\Big]^2,$$
where $\hat{Y}_i$ is the estimator of $Y_i$.

\begin{table}[hb]
\centering
\caption{Model-based simulation setup with definition of parameters and distributions} \label{table3}
\renewcommand{\arraystretch}{1.25}
\begin{tabular}{l}
\\\hline
\textbf{Simulation Setup:}\\
Generate $W_i$ from a Normal(5,9)
 and $\psi_i$ from a Gamma(4.5,2)\\
Take $\theta_i = 3 W_i+\nu_i$, $z_i=\theta_i+e_i$, and $w_i=W_i+\eta_i$\\ 
$\nu_i \sim \text{Normal}(0,\sigma^2_{\nu})$, 
$e_i \sim \text{Normal}(0,\psi_i)$, and $\eta_i \sim \text{Normal}(0,\Sigma_i)$\\
Take $y_i=\exp(z_i)$ and $Y_i=\exp(\theta_i)$\\
\hline\hline 
\textbf{Parameter Definition:} \\
Let $m=20,50,100,$ and $500$ (number of small areas)\\
Let  $\sigma^2_{\nu}=2$ (for all cases) \\
Let $k \in \{0, 20, 50, 80, \text{and} \; 100 \}$ \\
$\Sigma_i \in \{0, d\}$, where $d = 2$ \text{or} $4$ \\
Allow $k\%$ of the $\Sigma_i$'s randomly receive $d$
and the rest $0$.  \\
\hline
\end{tabular}
\end{table}

In order to evaluate the jackknife and parametric bootstrap estimators of $\hat{Y}_i^{\text{EB}}$, we consider the relative bias, denoted by $\text{RB}_J(\hat{Y}_i^{\text{EB}})$ and $\text{RB}_B(\hat{Y}_i^{\text{EB}})$, respectively. More specifically, the relative biases are defined as follows for each area $i$:
\begin{align*}
\text{RB}_J(\hat{Y}_i^{\text{EB}}) &= \Big\{\frac{1}{R} \sum_{r=1}^{R} \text{mspe}_J^{(r)} (\hat{Y}_i^{\text{EB}(r)}) - \text{EMSE}(\hat{Y}_i^{\text{EB}}) \Big\} \Big/\text{EMSE}(\hat{Y}_i^{\text{EB}}),  \\
\text{RB}_B(\hat{Y}_i^{\text{EB}}) &= \Big\{\frac{1}{R} \sum_{r=1}^{R} \text{mspe}_B^{(r)} (\hat{Y}_i^{\text{EB}(r)}) - \text{EMSE}(\hat{Y}_i^{\text{EB}}) \Big\} \Big/\text{EMSE}(\hat{Y}_i^{\text{EB}}).
\end{align*}

\subsubsection{Model-Based Simulation Results} \label{sec:model-results}
In this section, we summarize our results of the model-based simulation study.
\paragraph{Investigating the performance of the proposed estimators}
In this section, we investigate the performance of the proposed estimators.
Table \ref{table4} provides the four estimators given in Sec.~\ref{sec:setup} with their empirical MSEs, where we average the results over all the small areas and re-scale them using the logarithmic scale. When $k=0$, the MSE's for all EB predictors are the same since the term $\Sigma_i$ vanishes and $w_i$ is the same as $W_i$. Overall, as the value of $k$ increases, the empirical MSE increases for almost all predictors. We observe there are cases in which the EB predictors cannot outperform the direct estimators based on the simulation results.

Table \ref{table4} illustrates that there are cases in which the $\text{EMSE}(\hat{Y}_i^{\text{EB}})$ is larger than the $\text{EMSE}(y_i).$ In fact, the EB predictors cannot outperform the direct estimators due to propagated errors in the term $\bb^\top \Sigma_i \bb,$ which is present in the term $\gamma_i$ in the EB predictors through the simulations; see expression (\ref{eqn:EBpredictor}). Therefore, as the measurement error variance $\Sigma_i$ increases, we have shown that the MSE of our proposed EB predictors can also increase. This is the main point that one should notice when using a log-model with measurement error. In order to prevent such behavior, a further adjustment should be made to the EB predictors, which we discuss in Sec.~\ref{sec:dis}.

\begin{table}[h!]
\centering
\caption{Estimators and their empirical MSEs from model-based simulations. The results are averaged over all the small areas and re-scaled logarithmically.} \label{table4}
\renewcommand{\arraystretch}{1.25}
\begin{tabular}{@{} cccccccccc @{}}
\\\hline
m & k & $y_i$ & $\hat{Y}_i$ & $\tilde{Y}_i$ & $\hat{Y}_i^{\text{EB}}$ & $\text{EMSE}(y_i)$ & $\text{EMSE}(\hat{Y}_i)$ & $\text{EMSE}(\tilde{Y}_i)$ & $\text{EMSE}(\hat{Y}_i^{\text{EB}})$ \\

\specialrule{.1em}{.05em}{.05em} 
20 &0 & 46.766 & 50.626 & 50.626 & 50.626 & 102.088 & 111.09 & 111.09 & 111.09  \\
&20 & 53.054 & 50.139 & 49.229 & 49.864 & 115.974 & 109.338 & 104.398 & 108.425\\
&50 & 42.232 & 41.174 & 42.464 & 43.110 & 93.496 & 91.163 & 92.948 & 94.223 \\
&80 & 42.519 & 44.285 & 43.469 & 44.794 & 93.881 & 97.557 & 95.152 & 97.495\\
&100 & 44.682 & 41.81 & 45.073 & 46.331 & 99.13 & 92.161 & 99.938 & 102.48 \\
\hline
50&0 & 49.624 & 47.732 & 47.732 & 47.732 & 110.061 & 106.167 & 106.167 & 106.167 \\
&20 & 44.292 & 42.851 & 44.702 & 45.098 & 97.644 & 94.541 & 99.321 & 100.255 \\
&50 & 45.512 & 44.677 & 46.071 & 47.59 & 99.615 & 98.56 & 101.351 & 105.547\\
&80 & 42.703 & 41.773 & 45.289 & 46.469 & 94.339 & 93.268 & 101.068 & 103.615\\
&100 & 43.83 & 43.201 & 44.779 & 45.68 & 97.144 & 95.961 & 98.347 & 100.319 \\
\hline
100&0 & 42.635 & 42.241 & 42.241 & 42.241 & 94.625 & 92.802 & 92.802 & 92.802\\
&20 & 46.216 & 45.601 & 46.179 & 47.427 & 103.264 & 101.412 & 103.035 & 106.172\\
&50 & 50.93 & 45.982 & 49.347 & 48.519 & 113.343 & 103.08 & 109.718 & 108.214\\
&80 & 50.132 & 46.586 & 48.137 & 48.966 & 111.618 & 103.055 & 106.712 & 108.454\\
&100 & 44.925 & 44.711 & 48.009 & 49.031 & 100.996 & 100.764 & 107.472 & 109.515 \\
\hline
500 & 0 & 47.338 & 45.253 & 45.253 & 45.253 & 107.275 & 103.465 & 103.465 & 103.465 \\
& 20 & 46.382 & 45.369 & 47.575 & 47.652 & 104.607 & 102.867 & 107.896 &  108.169 \\
& 50 & 53.045 & 46.854 & 50.662 & 49.126 & 119.208 & 106.396 & 114.378 & 110.006 \\
& 80 & 47.766 & 44.868 & 47.706 & 49.950 & 108.289 & 104.921 & 107.454 &  112.805 \\
& 100 & 48.586 & 45.313 & 49.372 & 50.449 & 109.795 & 103.378 & 111.069 & 113.218 \\
\hline
\end{tabular}
\end{table}

\paragraph{Jackknife versus parametric bootstrap estimators}
We compare the jackknife MSPE estimator of the EB predictor $\hat{Y}^{\text{EB}}_i$ to that of the bootstrap using the relative bias (see Table \ref{table6}). 
In addition, we consider box plots for the jackknife and bootstrap MSPE estimators of the EB predictor $\hat{Y}^{\text{EB}}_i,$ where we compare these to box plots of the true values (see Figure \ref{figure3}). Both Table \ref{table6} and Figure \ref{figure3} illustrate that the bootstrap receives a large number of negative values, which is due to the construction of $\hat{M}_{1i}.$ 
Here, we find that the bootstrap grossly underestimates the true values, whereas the jackknife slightly overestimates the true values. This could be due to generating data from the normal distribution and the non-linear transformation in the model. Thus, we would recommend the jackknife in practice.

\paragraph{Amount of zeros for the estimates of $\hat{\sigma}_u^2$} Here, we investigate the proportion of zero estimates for $\sigma^2_{\nu}$ based on iteratively solving the Eqs. (\ref{eqn:bb}) and (\ref{eqn:sigma}). Figure \ref{figure2} illustrates that  as the number of small areas increases, the magnitude of receiving zeros decreases. More specifically, we observe when $m=20$ and as $k$ increases, $\hat{Y}_i^{\text{EB}}$ and $\hat{Y}_i$ tend to have a proportion of zero estimates of $\sigma^2_{\nu}$ between 0.3 and 0.5. 
When $m=50$ and as $k$ increases, $\hat{Y}_i^{\text{EB}}$ and $\hat{Y}_i$ tend to have a proportion of zero estimates of $\sigma^2_{\nu}$ between 0.15 and 0.4. 
When $m=100$ and as $k$ increases, $\hat{Y}_i^{\text{EB}}$ and $\hat{Y}_i$ tend to have a proportion of zero estimates of $\sigma^2_{\nu}$ between 0.05 and 0.3.
When $m=500$ and as $k$ increases, $\hat{Y}_i^{\text{EB}}$ and $\hat{Y}_i$ tend to have a proportion of zero estimates of $\sigma^2_{\nu}$ between 0 and 0.05.
One should be cautious of this in practical applications, and adjusting for this is of the interest of future work.

\begin{figure}[ht]
\centering
\begin{tabular}{c}
\includegraphics[width=0.75\textwidth]{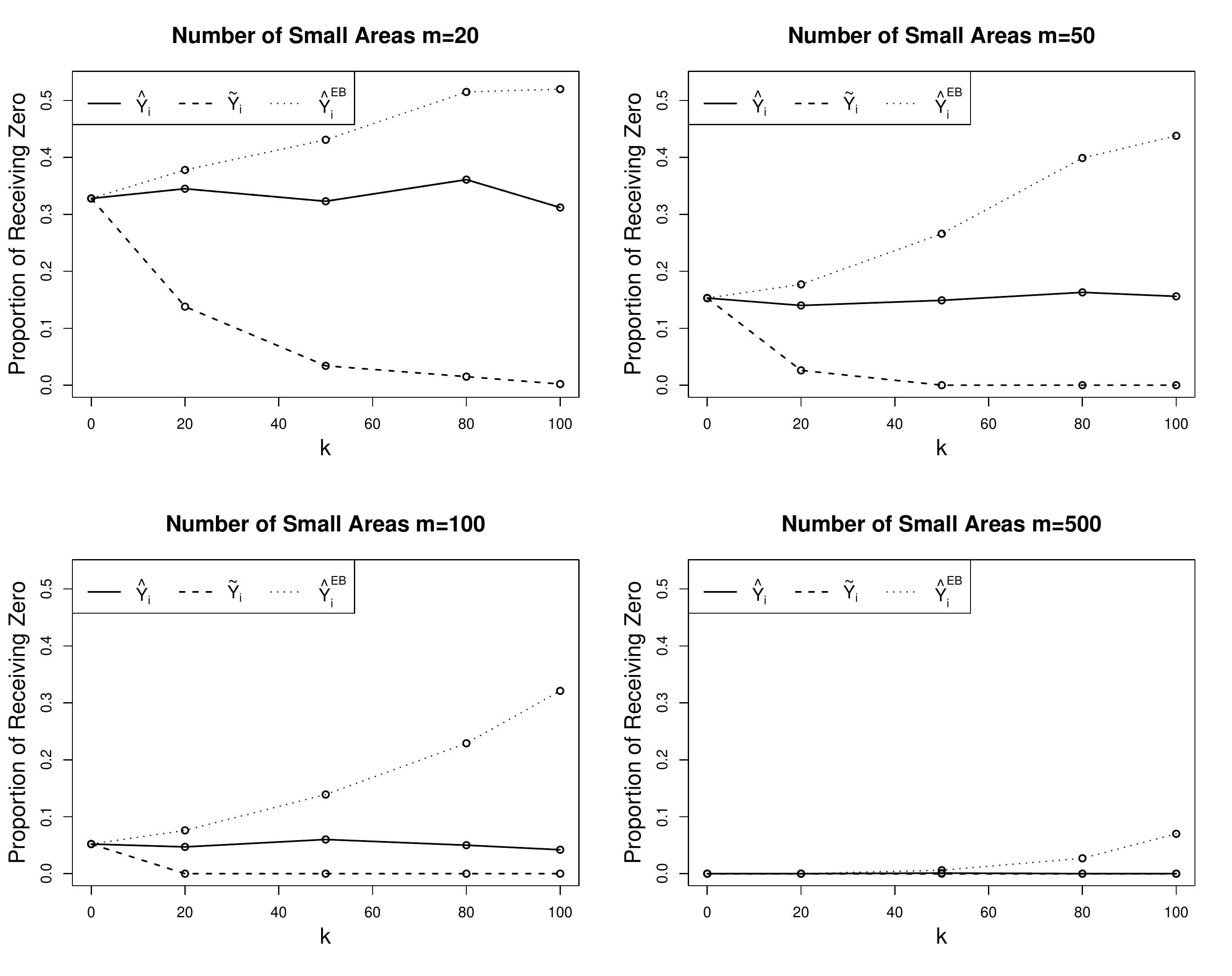}
\end{tabular}
\caption{The proportion of zero estimates of $\sigma^2_{\nu}$ from model-based simulation when we perform 1,000 replications of the simulation study for $k=0,\ldots100$, $m=20,50,100,500$, and $d=2$. 
}
\label{figure2}
\end{figure}

\paragraph{The effect of misspecification of $\Sigma_i$ on $\beta$} \label{sec:beta}
We investigate the effect of mis-specifying the variance $\Sigma_i$ on  the estimation of the regression parameter $\beta.$ To accomplish this, we 
conduct an empirical study based on the proposed model-based simulation study in Table \ref{table3} for the EB predictor $\hat{Y}^{\text{EB}}_i$. 
Assume $\beta=3,$ and we consider two sets of experiments for each value of $k,$ which are summarized in Table \ref{table7}. Recall that $\Sigma_i \in \{0,d \}.$
Denote the first set of experiments by 
$A1, B1, C1, D1$ and $E1,$ where we assume $d=2.$ Denote the misspecified value of $d$ by $d_{\text{mis}} = 4.$ Denote the second set of experiments by 
$A2, B2, C2, D2$ and $E2,$ where we assume $d=4$ and $d_{\text{mis}} = 2.$
We conduct both sets of experiments for $m=20$ and $500$.
For each experiment, we estimate the unknown parameter $\beta$ under the followings: 
(1) the true value of $d$ denoted by $\hat{\beta}$ and (2) the misspecified value of $d_{\text{mis}}$ denoted by $\hat{\beta}_{\text{mis}}$. 

Then we compute the average absolute difference between the respective $\beta$'s by considering the following:
$$ 100 \times \frac{1}{R}\sum\limits_{r=1}^{R}\Big|\hat{\beta}^{(r)}-\hat{\beta}_{\text{mis}}^{(r)}\Big|.$$ 
In addition, we compute the magnitude of bias related to $\hat{\beta}$ and $\hat{\beta}_{\text{mis}}$ with respect to the true value of $\beta=3$ as follows
$$100 \times \frac{1}{R} \sum_{r=1}^{R} \Big(\hat{\beta}^{(r)} -3 \Big) \quad \text{and} \quad 100 \times \frac{1}{R} \sum_{r=1}^{R} \Big( \hat{\beta}^{(r)}_{\text{mis}}-3\Big). $$

Table \ref{table7} illustrates that the overall misspecification of $\Sigma_i$ leads to bias in $\beta.$ When the magnitude of measurement error is zero (i.e. $k=0$), there is no difference between the estimated $\beta$ using $d$ or $d_{\text{mis}}$. On the other hand, when the magnitude of $k$ increases and we have more uncertainty in the error variance $\Sigma_i$, values of $\hat{\beta}$ and $\hat{\beta}_{\text{mis}}$ diverge more from one another, and the magnitude of the bias increases. Also, we observe as the number of small areas increases, the value of bias decreases.
One can resolve this bias issue by constructing an adaptive estimator for $\hat{\beta}_{\text{mis}}$ in which its bias is corrected through some techniques such as bootstrap and develop a test of parameter specification.

\begin{table}[ht]
\centering
\caption{Comparison of the proposed jackknife and bootstrap estimators from model-based simulations. The results are averaged over all small ares. The results are rescaled by the logarithm of absolute value. Note, “*” denote that the original value is negative.}

\label{table6}
\renewcommand{\arraystretch}{1.25}
\begin{tabular}{@{} ccccccc @{}}
\\\hline
m & k & $\text{EMSE}(\hat{Y}_i^{\text{EB}})$ & $\text{mspe}_J(\hat{Y}_i^{\text{EB}})$ & $\text{mspe}_B(\hat{Y}_i^{\text{EB}})$  & $\text{RB}_J(\hat{Y}_i^{\text{EB}})$ & $\text{RB}_B(\hat{Y}_i^{\text{EB}})$ \\\specialrule{.1em}{.05em}{.05em} 
20 & 0 & 111.09 & 120.731 & 118.478* & 9.641 & 7.389* \\
&20 &  108.425 & 115.222 & 115.884* & 6.796 & 7.459*\\
&50 & 94.223 & 107.245 & 111.255* & 13.021 & 17.032*\\
&80 & 97.495 & 113.06 & 129.772* & 15.565 & 32.277*\\
&100 & 102.48 & 115.536* & 119.84* & 13.056* & 17.36*\\
\hline
50 &0 & 106.167 & 112.318* & 116.383* & 6.154* & 10.216* \\
&20 & 100.255 & 110.449 & 106.743* & 10.194 & 6.489* \\
&50 &  105.547 & 115.365* & 118.838* & 9.818* & 13.291*\\
&80 & 103.615 & 114.127 & 112.717* & 10.512 & 9.102*\\
&100 & 100.319 & 110.454* & 112.552* & 10.134* & 12.233*\\
\hline
100&0 & 92.802 & 98.67 & 102.932* & 5.866 & 10.13*\\
&20 & 106.172 & 106.788* & 108.623* & 1.048* & 2.534* \\
&50 & 108.214 & 119.666 & 120.032* & 11.452 & 11.819*\\
&80 & 108.454 & 117.223* & 119.418* & 8.769* & 10.964*\\
&100 & 109.515 & 119.028 & 117.071* & 9.513 & 7.557*\\
\hline
500& 0 & 103.465 & 108.38 & 110.455* & 4.908 & 6.991* \\
& 20 &  108.169 & 119.672 & 115.892 & 11.502 & 7.722 \\
& 50 & 110.006 & 121.191 & 125.754* & 11.184 & 15.747*\\
& 80 & 112.805 & 125.672* & 129.761* & 12.867* & 16.955* \\
& 100 & 113.217 & 123.037* & 124.597* & 9.819* & 11.379* \\
\hline
\end{tabular}
\end{table}

\begin{figure}[ht]
\centering
\begin{tabular}{c}
\includegraphics[width=0.8\textwidth]{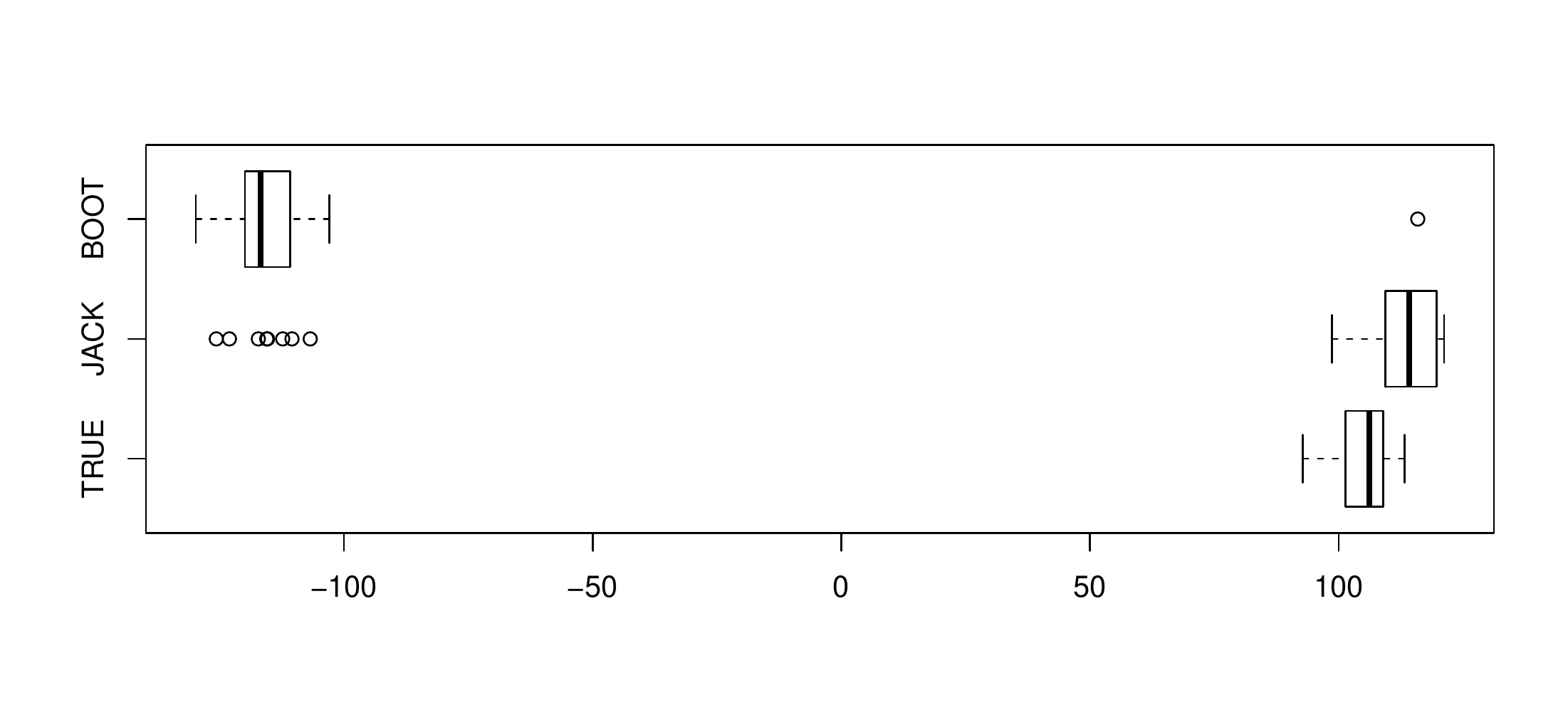}
\end{tabular}
\caption{Comparing the distribution of the jackknife and bootstrap estimators with respect to the true MSE ($\text{EMSE}(\hat{Y}_i^{\text{EB}})$) from the model-based simulations. The results are logarithmically rescaled.}
\label{figure3}
\end{figure}

\begin{table}[ht]
\centering
\caption{Percentage of bias related to the consequences of misspecifying the error variance $\Sigma_i$ on $\beta$ in the EB predictor $\hat{Y}^{\text{EB}}_i$ from model-based simulations. For all cases, we assume the true value for $\beta$ is 3. Also, $\sigma^2_{\nu}=2$ and $m=20$ (the smallest one) and $500$ (the largest one).} 
\label{table7}
\renewcommand{\arraystretch}{1.25}
\vspace{0.25cm}
\begin{tabular}{@{} cccccc @{}}
\hline
m & k & Experiment & $\frac{1}{R}\sum\limits_{r=1}^{R}\Big|\hat{\beta}^{(r)}-\hat{\beta}_{\text{mis}}^{(r)}\Big|$ & $\frac{1}{R}\sum\limits_{r=1}^{R}\Big(\hat{\beta}^{(r)}-3\Big)$ & $\frac{1}{R}\sum\limits_{r=1}^{R}\Big(\hat{\beta}_{\text{mis}}^{(r)}-3\Big)$\\ \specialrule{.1em}{.05em}{.05em} 
20 & $0$ & $A1 (d=2, d_{\text{mis}}=4)$ & $0$ & $-0.026$ & $-0.026$ \\
 \vspace{0.25cm}
 && $A2 (d=4, d_{\text{mis}}=2)$ & $0$ & $-0.237$ & $-0.237$\\ 
& $20$ & $B1 (d=2, d_{\text{mis}}=4)$ & $6.034$ & $0.314$ & $-0.046$\\
 \vspace{0.25cm}
 && $B2 (d=4, d_{\text{mis}}=2)$ & $5.814$ & $-0.634$ & $-0.591$\\ 
&$50$ & $C1 (d=2, d_{\text{mis}}=4)$ & $11.287$ & $-0.184$ & $-0.864$\\
 \vspace{0.25cm}
 && $C2 (d=4, d_{\text{mis}}=2)$ & $10.977$ & $-0.157$ & $-0.394$\\ 
&$80$ & $D1 (d=2, d_{\text{mis}}=4)$ & $19.641$ & $0.842$ & $2.177$\\
 \vspace{0.25cm}
 && $D2 (d=4, d_{\text{mis}}=2)$ & $18.722$ & $0.712$ & $0.327$\\ 
&$100$ & $E1 (d=2, d_{\text{mis}}=4)$ & $25.774$ & $2.650$ & $4.285$\\
 \vspace{0.25cm}
 && $E2 (d=4, d_{\text{mis}}=2)$ & $25.967$ & $4.020$ & $2.937$\\ 
 \hline

500 & $0$ & $A1 (d=2, d_{\text{mis}}=4)$ & $0$ & $0.185$ & $0.185$ \\
 \vspace{0.25cm}
 && $A2 (d=4, d_{\text{mis}}=2)$ & $0$ & $-0.082$ & $-0.082$\\ 
& $20$ & $B1 (d=2, d_{\text{mis}}=4)$ & $1.136$ & $0.174$ & $0.118$\\
 \vspace{0.25cm}
 && $B2 (d=4, d_{\text{mis}}=2)$ & $1.107$ & $0.010$ & $0.009$\\ 
&$50$ & $C1 (d=2, d_{\text{mis}}=4)$ & $2.200$ & $-0.138$ & $0.026$\\
 \vspace{0.25cm}
 && $C2 (d=4, d_{\text{mis}}=2)$ & $2.182$ & $-0.044$ & $0.006$\\ 
&$80$ & $D1 (d=2, d_{\text{mis}}=4)$ & $3.365$ & $0.204$ & $-0.067$\\
 \vspace{0.25cm}
 && $D2 (d=4, d_{\text{mis}}=2)$ & $3.386$ & $-0.139$ & $-0.090$\\ 
&$100$ & $E1 (d=2, d_{\text{mis}}=4)$ & $5.086$ & $0.174$ & $0.152$\\
 \vspace{0.25cm}
 && $E2 (d=4, d_{\text{mis}}=2)$ & $4.941$ & $-0.022$ & $0.117$\\ 
 \hline
\end{tabular}
\end{table}

\clearpage
\newpage
\section{Discussion}
\label{sec:dis}

In this paper, in order to stabilize the skewness and achieve normality in the response variable, we have proposed an area-level log-measurement error model on the response variable. In addition, we have proposed a  measurement error model on the covariates. Second, under our proposed modeling framework, we derived the EB predictor of positive small area quantities subject to the covariates containing measurement
error. Third, we proposed a corresponding estimate of MSPE using a jackknife and a bootstrap method, where we illustrated that the order of the bias is $O(m^{-1})$, where $m$ is the number of small areas.  Fourth, we have illustrated the performance of our methodology in both design-based simulation and model-based simulation studies, where the EMSE of the proposed EB predictor is not always uniformly better than that of the direct estimator. Our model-based simulation studies have provided further  investigation and guidance on the behavior. For example, one fruitful area of future research would be providing a correction to the EB predictor to avoid for such behavior. One way to address this issue is by estimating $\phi=(\bb,\sigma^2_{\nu})^\top$ in such a way that its order of bias is smaller than $O(m^{-1})$. This could help to reduce the amount of propagated errors in the EB predictor $\hat{Y}_i^{\text{EB}}$. Another way, which is less theoretical burdensome is by estimating the covariate as we have assumed it follows a functional measurement error model rather than a structural one. We have studied the MSPE of the  EB predictor using both the jackknife and bootstrap in simulation studies, where we have shown the jackknife estimator performs better than the bootstrap one under our log model.

\bibliographystyle{ims}
\bibliography{Bibliography}

\begin{thebibliography}{18}
\expandafter\ifx\csname natexlab\endcsname\relax\def\natexlab#1{#1}\fi
\expandafter\ifx\csname url\endcsname\relax
  \def\url#1{\texttt{#1}}\fi
\expandafter\ifx\csname urlprefix\endcsname\relax\def\urlprefix{URL }\fi
\providecommand{\eprint}[2][]{\url{#2}}

\bibitem[{Arima et~al.(2017)Arima, Bell, Datta and Liseo}]{arima2017}
{Arima, S.}, {Bell, W.}, {Datta, C., Gauri S.~Franco} and {Liseo, B.} 2017.
\newblock Multivariate fay--herriot bayesian estimation of small area means
  under functional measurement error.
\newblock \textit{Journal of the Royal Statistical Society: Series A,} 180 (4):
  1191--1209.

\bibitem[{Arima et~al.(2015{\natexlab{a}})Arima, Datta and
  Liseo}]{arima2015bayesianbook}
{Arima, S.}, {Datta, G.~S.} and {Liseo, B.} 2015{\natexlab{a}}.
\newblock \textit{Accounting for measurement error in covariates in SAE: an
  overview in analysis of poverty data by small area methods}.
\newblock M. Pratesi Editor, Wiley New York.

\bibitem[{Arima et~al.(2015{\natexlab{b}})Arima, Datta and
  Liseo}]{arima2015bayesian}
{Arima, S.}, {Datta, G.~S.} and {Liseo, B.} 2015{\natexlab{b}}.
\newblock Bayesian estimators for small area models when auxiliary information
  is measured with error.
\newblock \textit{Scandinavian Journal of Statistics,} 42 (2): 518--529.

\bibitem[{Bellow and Lahiri(2011)}]{bellow2011empirical}
{Bellow, M.~E.} and {Lahiri, P.~S.} 2011.
\newblock An empirical best linear unbiased prediction approach to small area
  estimation of crop parameters.
\newblock \textit{National Agricultural Statistics Service, University of
  Maryland}.

\bibitem[{Berg and Chandra(2014)}]{berg2014small}
{Berg, E.} and {Chandra, H.} 2014.
\newblock Small area prediction for a unit-level lognormal model.
\newblock \textit{Computational Statistics \& Data Analysis,} 78: 159--175.

\bibitem[{Butar and Lahiri(2003)}]{butar2003measures}
{Butar, F.~B.} and {Lahiri, P.} 2003.
\newblock On measures of uncertainty of empirical bayes small-area estimators.
\newblock \textit{Journal of Statistical Planning and Inference,} 112 (2):
  63--76.

\bibitem[{Chandra and Chambers(2011)}]{chandra2011small}
{Chandra, H.} and {Chambers, R.} 2011.
\newblock Small area estimation for skewed data in presence of zeros.
\newblock \textit{Calcutta Statistical Association Bulletin,} 63 (4): 241--258.

\bibitem[{Fay and Herriot(1979)}]{fay1979estimates}
{Fay, R.~E.} and {Herriot, R.~A.} 1979.
\newblock Estimates of income for small places: an application of james--stein
  procedures to census data.
\newblock \textit{Journal of the American Statistical Association,} 74 (366a):
  269--277.

\bibitem[{Fuller(2006)}]{fuller2006measurement}
{Fuller, W.~A.} 2006.
\newblock \textit{Measurement error models}.
\newblock John Wiley \& Sons.

\bibitem[{Ghosh et~al.(2015)Ghosh, Kubokawa and
  Kawakubo}]{ghosh2015benchmarked}
{Ghosh, M.}, {Kubokawa, T.} and {Kawakubo, Y.} 2015.
\newblock Benchmarked empirical bayes methods in multiplicative area-level
  models with risk evaluation.
\newblock \textit{Biometrika,} 102 (3): 647--659.

\bibitem[{Householder(1970)}]{householder1970numerical}
{Householder, A.~S.} 1970.
\newblock \textit{The numerical treatment of a single nonlinear equation}.
\newblock McGraw-Hill.

\bibitem[{Jiang et~al.(2002)Jiang, Lahiri and Wan}]{jiang2002unified}
{Jiang, J.}, {Lahiri, P.} and {Wan, S.~M.} 2002.
\newblock A unified jackknife theory for empirical best prediction with
  m-estimation.
\newblock \textit{The Annals of Statistics,} 30 (6): 1782--1810.

\bibitem[{Lohr and Rao(2009)}]{lohr2009jackknife}
{Lohr, S.~L.} and {Rao, J.} 2009.
\newblock Jackknife estimation of mean squared error of small area predictors
  in nonlinear mixed models.
\newblock \textit{Biometrika}, \textbf{96} 96: 457--468.

\bibitem[{Pfefferman(2013)}]{pfeffermann2013new}
{Pfefferman, D.} 2013.
\newblock New important developments in small area estimation.
\newblock \textit{Statistical Science} 28 (1): 40--68.

\bibitem[{Rao and Molina(2015)}]{rao2015small}
{Rao, J.~N.} and {Molina, I.} 2015.
\newblock \textit{Small area Estimation}.
\newblock Wiley Series in Survey Sampling, London.

\bibitem[{Slud and Maiti(2006)}]{slud2006mean}
{Slud, E.~V.} and {Maiti, T.} 2006.
\newblock Mean-squared error estimation in transformed fay--herriot models.
\newblock \textit{Journal of the Royal Statistical Society: Series B
  (Statistical Methodology),} 68 (2): 239--257.

\bibitem[{Wang and Fuller(2003)}]{wang2003mean}
{Wang, J.} and {Fuller, W.~A.} 2003.
\newblock The mean squared error of small area predictors constructed with
  estimated area variances.
\newblock \textit{Journal of the American Statistical Association,} 98 (463):
  716--723.

\bibitem[{Ybarra and Lohr(2008)}]{ybarra2008small}
{Ybarra, L.~M.} and {Lohr, S.~L.} 2008.
\newblock Small area estimation when auxiliary information is measured with
  error.
\newblock \textit{Biometrika,} 95 (4): 919--931.

\end{thebibliography}

\end{document}